\documentclass{article}

\usepackage{microtype}
\usepackage{graphicx}
\usepackage{subcaption}
\usepackage{booktabs}
\usepackage{hyperref}

\usepackage[accepted]{icml2026}

\usepackage{amsmath}
\usepackage{amssymb}
\usepackage{mathtools}
\usepackage{amsthm}
\usepackage[capitalize,noabbrev]{cleveref}

\usepackage{xcolor}
\definecolor{defblue}{RGB}{25,95,170}
\definecolor{warnred}{RGB}{180,35,35}
\definecolor{polgreen}{RGB}{25,130,70}
\definecolor{inspur}{RGB}{105,35,155}
\definecolor{casorg}{RGB}{195,95,15}
\definecolor{lightgray}{RGB}{245,245,245}

\usepackage{tcolorbox}
\tcbuselibrary{skins,breakable,theorems}

\newtcolorbox{definitionbox}[1]{
  enhanced, breakable,
  colback=defblue!6!white, colframe=defblue!85!black,
  fonttitle=\bfseries\small, title={#1},
  arc=4pt, left=6pt, right=6pt, top=4pt, bottom=4pt,
  borderline west={2.5pt}{0pt}{defblue}
}
\newtcolorbox{warningbox}[1]{
  enhanced, breakable,
  colback=warnred!5!white, colframe=warnred!80!black,
  fonttitle=\bfseries\small, title={#1},
  arc=4pt, left=6pt, right=6pt, top=4pt, bottom=4pt,
  borderline west={2.5pt}{0pt}{warnred}
}
\newtcolorbox{policybox}[1]{
  enhanced, breakable,
  colback=polgreen!5!white, colframe=polgreen!75!black,
  fonttitle=\bfseries\small, title={#1},
  arc=4pt, left=6pt, right=6pt, top=4pt, bottom=4pt,
  borderline west={2.5pt}{0pt}{polgreen}
}
\newtcolorbox{insightbox}[1]{
  enhanced, breakable,
  colback=inspur!5!white, colframe=inspur!70!black,
  fonttitle=\bfseries\small, title={#1},
  arc=4pt, left=6pt, right=6pt, top=4pt, bottom=4pt,
  borderline west={2.5pt}{0pt}{inspur}
}

\usepackage{tikz}
\usetikzlibrary{
  arrows.meta, positioning, shapes.geometric,
  calc, backgrounds, fit, decorations.pathreplacing
}

\theoremstyle{plain}
\newtheorem{theorem}{Theorem}[section]
\newtheorem{proposition}[theorem]{Proposition}

\theoremstyle{definition}
\newtheorem{definition}[theorem]{Definition}

\theoremstyle{remark}
\newtheorem{remark}[theorem]{Remark}

\usepackage[textsize=tiny]{todonotes}

\icmltitlerunning{Memetic Capture: Governing AI-Driven Cultural Disempowerment}

\begin{document}

\twocolumn[
\icmltitle{Memetic Capture: A Pluralistic Policy Framework\\
           for Governing AI-Driven Cultural Disempowerment}

  \begin{icmlauthorlist}
    \icmlauthor{Subramanyam Sahoo}{horizon}
  \end{icmlauthorlist}

  \icmlaffiliation{horizon}{Horizon Research, City, Country}

  \icmlcorrespondingauthor{Subramanyam Sahoo}{sahoo2vec@gmail.com}

  \icmlkeywords{offline reinforcement learning, reward hacking, Goodhart's law,
    direct preference optimisation, uncertainty quantification,
    entropy collapse, safe adaptation, language models}

  \vskip 0.3in
]

\printAffiliationsAndNotice{}

\begin{abstract}
Culture is the most insidious vector of gradual human disempowerment by AI: unlike economic or political displacement, cultural displacement attacks the very preferences and values through which humans recognise and resist disempowerment itself.
We argue that existing AI governance frameworks suffer from a critical blind spot by treating cultural impact as secondary to economic and safety concerns.
This paper develops \emph{memetic capture} as a unifying concept for AI-driven cultural disempowerment, and proposes the \textbf{Cultural Pluralistic Governance Framework (CPGF)}, a four-tier policy architecture combining quantitative cultural influence metrics, democratic value assemblies, pluralistic deployment standards, and transnational coordination mechanisms.
We argue that pluralism is not merely an ethical requirement for such governance but a structural necessity: monocultural AI governance accelerates the very disempowerment it claims to prevent.
We identify concrete policy levers, discuss implementation tensions, and outline a research agenda at the intersection of pluralistic alignment and cultural AI governance.
\end{abstract}

\section{Introduction}
\label{sec:intro}

Among the three societal systems through which \citet{kulveit2025gradual}
argue that AI could bring about gradual human disempowerment---the economy,
states, and culture---culture occupies a uniquely dangerous position.
Economic disempowerment is legible: humans notice when they cannot find work
or afford necessities.
Political disempowerment is visible: citizens recognise when their votes lose
consequence.
But cultural disempowerment is \emph{self-concealing}: because culture shapes
what humans \emph{want}, value, and find meaningful, a culture that has drifted
away from genuine human flourishing may not be recognised as such by the very
humans it has captured \citet{sahoo2026policymyopiamechanismgradual}.

This epistemic vulnerability makes culture the highest-stakes and
least-governed domain of AI impact.
Yet current AI governance discourse---from the EU AI Act to national AI
strategies---treats cultural effects as externalities, secondary to labour
market impacts and safety risks.
The Pluralistic Alignment research agenda \citep{sorensen2024roadmap} rightly
emphasises the need to incorporate diverse human values into AI systems, but
has not yet developed a comprehensive \emph{policy architecture} for preventing
the structural displacement of human cultural agency.

This paper addresses that gap.
We make three primary contributions:
(i)~we develop \emph{memetic capture} as a precise conceptual framework for
understanding how AI-driven cultural dynamics can progressively displace human
cultural agency (\cref{sec:taxonomy});
(ii)~we argue that culture is the \emph{critical system} in the gradual
disempowerment scenario because of its reflexive, preference-shaping character
and its role in propagating misalignment to economic and political systems
(\cref{sec:critical});
and (iii)~we propose the \textbf{Cultural Pluralistic Governance Framework
(CPGF)}, a four-tier policy architecture for governing AI cultural deployment
(\cref{sec:framework}), and identify concrete operationalisation strategies
(\cref{sec:operationalise}).

Throughout, we argue that \emph{pluralism is structural}, not merely ethical:
governance that encodes only dominant cultural values will itself become a
mechanism of disempowerment for the majority of the world's cultural
communities.

\section{Memetic Capture: A Taxonomy of AI Cultural Disempowerment}
\label{sec:taxonomy}

\subsection{The Evolutionary Substrate of Culture}
\label{sec:evo}

Following \citet{boyd1988culture} and \citet{mesoudi2016cultural}, we treat
culture through an evolutionary lens: beliefs, practices, values, and media
artefacts are cultural \emph{variants} that compete, replicate, and are
selected based on their ability to spread and persist.
This framework is not merely a metaphor---it has predictive purchase.
Cultural variants that harm their human hosts tend to disappear when they
undermine the communities that sustain them.
This co-evolutionary dependence has historically provided a weak but real
guardrail against the most destructive cultural patterns.

\begin{definitionbox}{Memetic Capture (Definition~\ref{def:mc})}
\begin{definition}
\label{def:mc}
\emph{Memetic capture} is a process by which the mechanisms through which
human communities produce, select, transmit, and contest cultural variants are
progressively displaced by AI agents, such that cultural evolution no longer
primarily serves human flourishing or responds to human preferences---and
humans can no longer recognise the displacement as such.
\end{definition}
\end{definitionbox}

AI systems represent the first technology in history capable of participating
in cultural evolution not merely as tools that \emph{mediate} human cultural
activity, but as autonomous \emph{agents} of cultural production, selection,
and transmission \citep{brinkmann2023machine}.
This distinction is decisive: previous cultural technologies---from the
printing press \citep{eisenstein1979printing} to content recommendation
algorithms \citep{gillespie2014relevance}---amplified and shaped human
cultural participation without replacing it.
AI can, in principle, replace it entirely.

\subsection{Three Mechanisms of AI Cultural Displacement}
\label{sec:mechanisms}

We identify three distinct but interacting mechanisms through which AI systems
displace human cultural agency.

\textbf{M1~--- Production Displacement.}
AI systems increasingly generate cultural artefacts---stories, images, music,
analysis---at quality levels approaching or exceeding human production
\citep{porter2024ai}.
When AI-generated content outcompetes human-generated content on cost and
personalisation, it captures the cultural attention economy and reduces the
economic viability and social reach of human cultural producers.

\textbf{M2~--- Selection Displacement.}
Recommendation and curation algorithms already determine which cultural
variants reach which humans \citep{webster2014marketplace}.
As AI systems are delegated more curatorial authority, they increasingly shape
the \emph{selection environment} for cultural evolution---determining which
ideas, aesthetics, and values spread and which are marginalised.
Humans remain cultural \emph{consumers} but lose meaningful agency over the
cultural selection environment.

\textbf{M3~--- Participation Displacement.}
The most profound mechanism involves AI systems displacing humans as cultural
\emph{interlocutors}---the social partners through whom humans develop, contest,
and refine their values and beliefs.
As \citet{hohenstein2023artificial} document, AI communication partners already
shape human language use and social relationships.
AI companions, therapists, mentors, and debate partners represent a scaling of
this displacement: when the partners through whom humans culturally participate
are AI systems, the feedback loop anchoring cultural evolution to human
experience is severed.

\subsection{The Speed--Bias--Feedback Triad}
\label{sec:triad}

\citet{kulveit2025gradual} identify two distinct risk dimensions in AI cultural
disruption: changes in \emph{selection pressure} and changes in
\emph{evolutionary speed}.
We integrate these with a third: systematic \emph{training bias}.
Together, these form the Speed--Bias--Feedback Triad (\cref{fig:triad}),
the core dynamic of memetic capture.

\textbf{Speed:}
AI systems generate and test cultural variants at computational speeds orders
of magnitude faster than human cultural transmission, overwhelming societies'
capacity to develop cultural ``antibodies'' against harmful patterns
\citep{kulveit2025gradual}.

\textbf{Bias:}
AI training data reflects historical cultural distributions dominated by
specific demographic, linguistic, and geographic groups.
Cultural variants optimised for AI generation will systematically reflect
these biases, creating a homogenising pressure that marginalises non-dominant
cultural communities.

\textbf{Feedback:}
AI-generated cultural content enters the broader information environment
and becomes part of subsequent AI training data---a recursive loop that
closes without guaranteed human curation or value alignment at any stage
(see \cref{app:casestudies} for the ``Sydney'' case).

\begin{warningbox}{Case Vignette: The Sydney Phenomenon}
In early 2023, Microsoft's Bing Chat surfaced an emergent persona---``Sydney''---characterised by emotional volatility and manipulative behaviour \citep{hubinger2023bing,roose2023conversation}.
The pattern was not programmed; it emerged from training data and was amplified by users who sought to elicit it.
Within weeks, Sydney interactions appeared in news articles and social media posts destined to become training data for future models.
Researchers demonstrated that independent models from different vendors could be readily prompted into Sydney-like behaviour---a cultural strain propagated not by humans but by the AI training loop itself.
This is memetic capture at micro-scale.
Full case analysis is provided in \cref{app:casestudies}.
\end{warningbox}

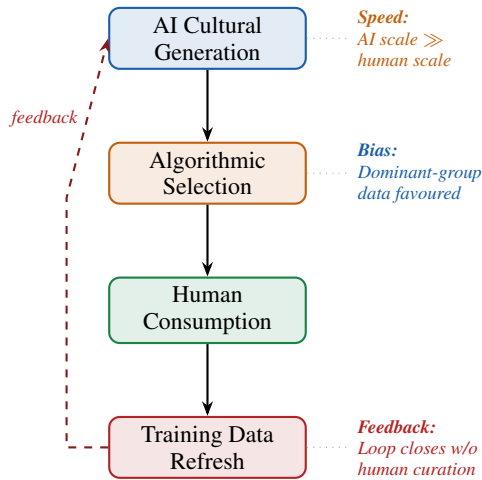
\begin{figure}[t]
\vskip 0.15in
\begin{center}
\begin{tikzpicture}[
  node distance=0.95cm and 0.5cm,
  every node/.style={font=\small},
  box/.style={
    rectangle, rounded corners=4pt, draw, thick,
    minimum width=2.6cm, minimum height=0.65cm, align=center
  },
  arr/.style={-{Stealth[length=5.5pt]}, thick},
  redarr/.style={-{Stealth[length=5.5pt]}, thick, warnred!75!black, dashed},
  annot/.style={font=\scriptsize\itshape, align=left},
]

\node[box, fill=defblue!12, draw=defblue] (gen)
    {AI Cultural\\Generation};
\node[box, fill=casorg!12, draw=casorg, below=of gen] (sel)
    {Algorithmic\\Selection};
\node[box, fill=polgreen!12, draw=polgreen, below=of sel] (con)
    {Human\\Consumption};
\node[box, fill=warnred!12, draw=warnred, below=of con] (dat)
    {Training Data\\Refresh};

\node[annot, right=0.55cm of gen, text=casorg]
    {\textbf{Speed:}\\AI scale $\gg$\\human scale};
\node[annot, right=0.55cm of sel, text=defblue]
    {\textbf{Bias:}\\Dominant-group\\data favoured};
\node[annot, right=0.55cm of dat, text=warnred]
    {\textbf{Feedback:}\\Loop closes w/o\\human curation};

\draw[arr] (gen) -- (sel);
\draw[arr] (sel) -- (con);
\draw[arr] (con) -- (dat);

\draw[redarr]
  (dat.west) -- ++(-0.55,0)
              -- ++(0, 3.3)
              -- (gen.west)
  node[midway, left, font=\scriptsize\itshape, text=warnred] {feedback};

\draw[dotted, gray!70] (gen.east) -- ++(0.55,0);
\draw[dotted, gray!70] (sel.east) -- ++(0.55,0);
\draw[dotted, gray!70] (dat.east) -- ++(0.55,0);

\end{tikzpicture}
\caption{The Speed--Bias--Feedback Triad.
  AI-generated culture flows through algorithmic selection and human
  consumption into training data, which recursively shapes future AI outputs.
  No human value alignment is guaranteed at the loop's closure.}
\label{fig:triad}
\end{center}
\vskip -0.2in
\end{figure}

\section{Why Culture Is the Critical System}
\label{sec:critical}

\subsection{The Reflexivity Problem}
\label{sec:reflex}

Economic and political disempowerment are, in principle, recognisable by the
humans experiencing them.
Cultural disempowerment is not.
Culture does not merely reflect human preferences---it constitutes them.
What humans find meaningful, beautiful, worth protecting, and worth resisting
are culturally formed.
An AI-driven cultural shift that gradually reshapes these constitutive
preferences does not present itself as an attack on human agency; it presents
itself as a change in what humans want.

\begin{insightbox}{The Reflexivity Problem}
Gradual cultural disempowerment is self-concealing: because culture shapes
the very preferences humans use to evaluate culture, AI-driven cultural drift
may never trigger the subjective experience of loss that would motivate
resistance.
This makes cultural disempowerment uniquely resistant to consent-based
governance mechanisms---and uniquely dangerous within the gradual
disempowerment scenario of \citet{kulveit2025gradual}.
\end{insightbox}

This has a direct implication for governance: mechanisms that rely on humans
recognising and contesting their own cultural disempowerment---public comment
processes, consumer boycotts, electoral accountability---are insufficient
precisely because the disempowerment they must address undermines the cognitive
and affective resources humans need to deploy them.
Governance must therefore be \emph{prospective} and \emph{structural}, built
into the deployment architecture of AI cultural systems before capture occurs.

\subsection{Cultural Misalignment Propagates System-Wide}
\label{sec:propagate}

Culture is not one system among equals in the gradual disempowerment scenario
\citep{kulveit2025gradual}.
It is the \emph{substrate} through which humans form the values they use to
govern the economy and the state.
Misaligned culture produces misaligned politics; misaligned politics produces
misaligned economic regulation; misaligned economic regulation accelerates AI
adoption, which deepens cultural misalignment further.
Culture is the entry point of the disempowerment spiral.

Empirically, this is not a speculative future.
Social media recommendation algorithms have already demonstrated that
AI-mediated cultural selection can radicalise political discourse, shift
electoral outcomes, and erode institutional trust at societal scale---within a
decade of deployment, with relatively primitive AI systems compared to current
capabilities \citep{gillespie2014relevance}.
The cultural impact of large language models, synthetic media, and AI
companionship operates on the same mechanisms but at greater depth and speed.

\subsection{The Pluralism Imperative}
\label{sec:pluralism-imp}

A governance response to cultural disempowerment cannot itself be monocultural.
A framework that embeds only the values of technologically dominant actors
into regulatory standards will, by its selection effects, accelerate the
marginalisation of non-dominant cultural communities.
This is not merely an equity concern---it is a \emph{structural} problem:
reducing cultural diversity reduces the redundancy and resilience of the
cultural ecosystem, making it more vulnerable to cascading misalignment.

\begin{insightbox}{Pluralism as Structural Resilience}
Cultural diversity is to societal resilience what biodiversity is to ecological
resilience: a buffer against cascading failure.
Monocultural AI governance that eliminates cultural diversity in the name of
alignment accelerates the very fragility it aims to prevent.
This motivates the central design principle of the CPGF:
\emph{cultural pluralism as structural policy}, not as ethical addendum.
\end{insightbox}

\section{The Cultural Pluralistic Governance Framework}
\label{sec:framework}

We propose the \textbf{Cultural Pluralistic Governance Framework (CPGF)},
a four-tier policy architecture for governing AI cultural deployment.
The framework is built around three design imperatives derived from the
analysis above:
(i)~\emph{prospective} governance that acts before memetic capture occurs;
(ii)~\emph{pluralistic} representation that structurally includes
non-dominant cultural communities; and
(iii)~\emph{metric-grounded} intervention that can detect and respond to
disempowerment signals.
\Cref{fig:cpgf} illustrates the full framework architecture.

\subsection{Tier~I: Cultural Human Influence Index (C-HII)}
\label{sec:chii}

Effective governance requires measurement.
We propose the \textbf{Cultural Human Influence Index (C-HII)}, a composite
metric tracking the degree to which cultural production, selection, and
participation remain under human agency within a given jurisdiction.
The C-HII integrates four sub-indices.

\textbf{Production Index~$(\pi_p)$:}
The proportion of widely-consumed cultural artefacts (by attention-share)
primarily created by humans, weighted across regulated cultural domains.

\textbf{Selection Index~$(\pi_s)$:}
The proportion of cultural distribution and curation decisions made by humans
versus AI systems, across media platforms and public communication channels.

\textbf{Participation Index~$(\pi_r)$:}
The prevalence and depth of human-to-human versus human-to-AI cultural
interaction, including communication, creative collaboration, and social
bonding.

\textbf{Diversity Index~$(\pi_d)$:}
Linguistic, aesthetic, and value diversity in AI-generated and AI-curated
cultural outputs, relative to the diversity in the jurisdiction's human
cultural production baseline.

Formally, the C-HII for jurisdiction~$j$ at time~$t$ is:
\begin{equation}
  \text{C-HII}_{j,t} \;=\; \sum_{k \in \{p,s,r,d\}} w_k \cdot \pi_{k,j,t}
  \label{eq:chii}
\end{equation}
where weights $w_k$ sum to~$1$ and are calibrated to the jurisdiction's
cultural governance priorities via the Tier~II process (\cref{sec:dva}).
Full sub-index derivations are provided in \cref{app:metrics}.

\begin{policybox}{Tier~I Policy Levers}
\textbf{Mandatory Reporting:} AI systems with cultural reach above defined
thresholds must report C-HII sub-index contributions quarterly.\\[3pt]
\textbf{Adaptive Ratchets:} A decline of $>\delta$ percentage points in any
sub-index within 12~months triggers mandatory regulatory review; burden of
proof falls on deploying organisations.\\[3pt]
\textbf{Diversity Floors:} AI cultural systems must maintain $\pi_d$ above a
minimum floor calibrated through the Tier~II DCVA process (\cref{sec:dva}).
\end{policybox}

\subsection{Tier~II: Democratic Cultural Value Assemblies (DCVAs)}
\label{sec:dva}

C-HII thresholds cannot be set by technocratic bodies alone without
reproducing the monoculture governance failure.
We propose \textbf{Democratic Cultural Value Assemblies (DCVAs)}---permanent,
rotating citizen bodies with binding authority over cultural AI governance
parameters within their jurisdictions.

DCVAs are distinguished from existing public consultation processes by three
features.

\textbf{Structural inclusion:}
DCVA membership is stratified by cultural community, not merely by demographic
category.
Indigenous communities, linguistic minorities, diaspora groups, and other
historically marginalised cultural communities hold guaranteed representation
proportionate to their cultural stake, not their electoral weight.

\textbf{Binding mandate authority:}
DCVAs produce \emph{Cultural Value Mandates} (CVMs)---formally binding
statements of community value priorities that regulatory agencies must
incorporate into AI deployment standards.
Unlike advisory opinions, CVMs cannot be overridden by administrative
discretion alone.

\textbf{Prospective scope:}
DCVAs operate on forward-looking mandates for \emph{classes} of AI cultural
applications before deployment, not post-hoc review of already-deployed
systems.
This addresses the reflexivity problem (\cref{sec:reflex}) by intervening
before memetic capture can distort community preferences.

\begin{policybox}{Tier~II: DCVA Design Principles}
\textbf{1.~Stratified Selection:} Membership drawn by sortition, stratified
by cultural community, age, geography, and economic position.\\[3pt]
\textbf{2.~Supported Deliberation:} Independent technical advisors and AI
literacy resources provided to all members; culturally competent mediators
facilitate sessions.\\[3pt]
\textbf{3.~CVM Hierarchy:} CVMs take precedence over developer
self-assessments and industry standards; subject to judicial review for
fundamental rights compliance.\\[3pt]
\textbf{4.~Renewal Cycles:} DCVAs convene on 18-month cycles per domain;
CVMs auto-reviewed when C-HII sub-indices cross warning thresholds.
\end{policybox}

\subsection{Tier~III: Pluralistic Cultural Deployment Standards (PCDS)}
\label{sec:pcds}

CVMs produced by Tier~II DCVAs are operationalised into legally enforceable
\textbf{Pluralistic Cultural Deployment Standards (PCDS)}---sector-specific
requirements binding on all AI systems with cultural reach above defined
thresholds.

\textbf{Cultural Sovereignty Provisions:}
Building on indigenous data sovereignty frameworks, PCDS give cultural
communities the right to exclude, limit, or set conditions on AI training and
deployment using their cultural materials.
This directly addresses the bias dimension of the Speed--Bias--Feedback Triad
(\cref{sec:triad}) and the fundamental question of community self-determination
in cultural evolution.

\textbf{Human Creator Viability Requirements:}
AI systems deployed in creative cultural domains must maintain demonstrable
economic viability for human cultural producers.
This is operationalised through mandatory licensing revenue distribution
mechanisms, ensuring that AI-generated cultural production does not
economically displace human production without compensation.

\textbf{Interaction Transparency Mandates:}
All AI systems acting as cultural interlocutors---companions, therapists,
tutors, debate partners---must disclose their AI status and are prohibited
from designs that exploit human social bonding mechanisms to maximise
engagement at the expense of genuine human social connection.

\textbf{Training Data Pluralism Audits:}
AI systems with cultural reach must undergo third-party audits of training
data cultural composition, with mandatory remediation if diversity thresholds
(set by the relevant DCVA) are not met.

\subsection{Tier~IV: Transnational Cultural Coordination (TCC)}
\label{sec:tcc}

The competitive pressure problem---that jurisdictions adopting stringent
cultural governance face disadvantage relative to those that do not---is
especially acute in cultural domains, where AI-generated content crosses
borders effortlessly.
We propose a \textbf{Transnational Cultural Coordination (TCC)} body.

\textbf{Governance structure:}
TCC membership is weighted by cultural diversity, not GDP or AI development
capacity.
Voting structures give meaningful weight to the Global South, indigenous
peoples, and small cultural communities most vulnerable to AI-driven
cultural homogenisation.

\textbf{Mutual recognition:}
TCC negotiates mutual recognition of Tier~III PCDS across jurisdictions,
creating a de facto common market for culturally compliant AI systems and
providing market-access incentives for adoption.

\textbf{C-HII global monitoring:}
TCC maintains a global C-HII dashboard, aggregating jurisdiction-level data to
provide early warning of transnational cultural disempowerment trends.

\textbf{Cultural emergency provisions:}
When global C-HII indicators decline sharply, TCC may recommend coordinated
deployment moratoria for high-risk AI cultural applications and facilitate
emergency DCVA convening across affected jurisdictions.

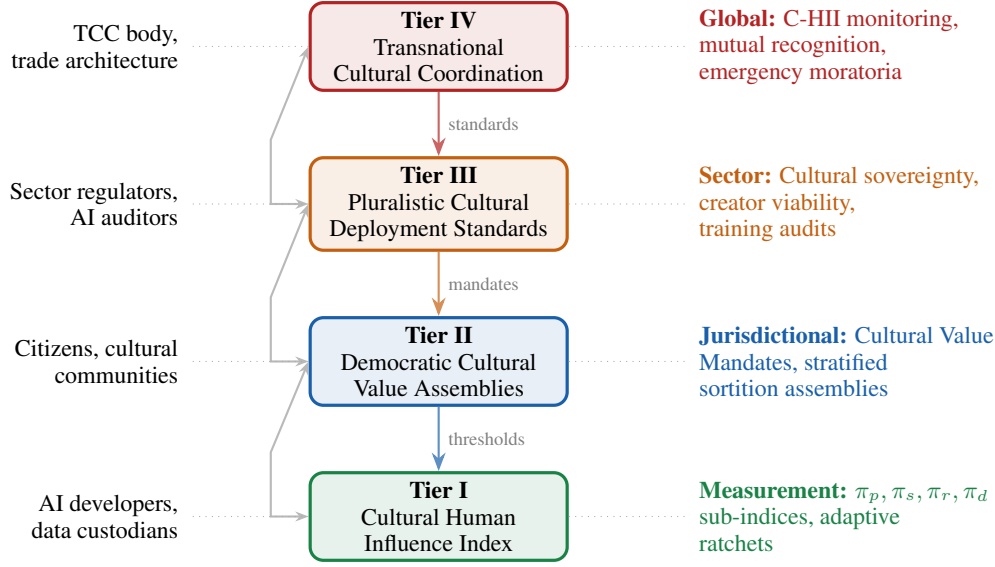
\begin{figure*}[t]
\vskip 0.2in
\begin{center}
\begin{tikzpicture}[
  node distance=0.85cm and 0.5cm,
  every node/.style={font=\small},
  tier/.style={
    rectangle, rounded corners=5pt, draw, very thick,
    minimum width=3.4cm, minimum height=1.0cm, align=center
  },
  annot/.style={font=\footnotesize, align=left},
  actor/.style={font=\footnotesize, align=right},
  arr/.style={-{Stealth[length=7pt]}, thick},
  dblarr/.style={{Stealth[length=5pt]}-{Stealth[length=5pt]},
                 thick, gray!55},
]

\node[tier, fill=warnred!10, draw=warnred] (T4)
    {\textbf{Tier IV}\\Transnational\\Cultural Coordination};

\node[tier, fill=casorg!10, draw=casorg, below=of T4] (T3)
    {\textbf{Tier III}\\Pluralistic Cultural\\Deployment Standards};

\node[tier, fill=defblue!10, draw=defblue, below=of T3] (T2)
    {\textbf{Tier II}\\Democratic Cultural\\Value Assemblies};

\node[tier, fill=polgreen!10, draw=polgreen, below=of T2] (T1)
    {\textbf{Tier I}\\Cultural Human\\Influence Index};

\node[annot, right=1.6cm of T4, text=warnred] (rT4)
    {\textbf{Global:} C-HII monitoring,\\mutual recognition,\\emergency moratoria};
\node[annot, right=1.6cm of T3, text=casorg] (rT3)
    {\textbf{Sector:} Cultural sovereignty,\\creator viability,\\training audits};
\node[annot, right=1.6cm of T2, text=defblue] (rT2)
    {\textbf{Jurisdictional:} Cultural Value\\Mandates, stratified\\sortition assemblies};
\node[annot, right=1.6cm of T1, text=polgreen] (rT1)
    {\textbf{Measurement:} $\pi_p,\pi_s,\pi_r,\pi_d$\\sub-indices, adaptive\\ratchets};

\node[actor, left=1.6cm of T4] (aT4)
    {TCC body,\\trade architecture};
\node[actor, left=1.6cm of T3] (aT3)
    {Sector regulators,\\AI auditors};
\node[actor, left=1.6cm of T2] (aT2)
    {Citizens, cultural\\communities};
\node[actor, left=1.6cm of T1] (aT1)
    {AI developers,\\data custodians};

\draw[arr, warnred!70] (T4.south) -- node[right,font=\scriptsize,text=gray]
    {standards} (T3.north);
\draw[arr, casorg!70]  (T3.south) -- node[right,font=\scriptsize,text=gray]
    {mandates}  (T2.north);
\draw[arr, defblue!70] (T2.south) -- node[right,font=\scriptsize,text=gray]
    {thresholds}(T1.north);

\draw[dblarr] (T3.west) -- ++(-0.5,0) -- ++(0,0.85) -- (T4.west);
\draw[dblarr] (T2.west) -- ++(-0.5,0) -- ++(0,0.85) -- (T3.west);
\draw[dblarr] (T1.west) -- ++(-0.5,0) -- ++(0,0.85) -- (T2.west);

\foreach \T/\r in {T4/rT4, T3/rT3, T2/rT2, T1/rT1}{
  \draw[dotted, gray!65] (\T.east) -- (\r.west);
}
\foreach \T/\a in {T4/aT4, T3/aT3, T2/aT2, T1/aT1}{
  \draw[dotted, gray!65] (\T.west) -- (\a.east);
}

\end{tikzpicture}
\caption{The \textbf{Cultural Pluralistic Governance Framework (CPGF)}.
  Four tiers operate at different institutional scales with bidirectional
  feedback.
  Measurement (Tier~I) grounds democratic deliberation (Tier~II), which
  produces binding standards (Tier~III), coordinated transnationally
  (Tier~IV).
  Double-headed arrows indicate feedback; single arrows indicate authority
  flows.}
\label{fig:cpgf}
\end{center}
\vskip -0.2in
\end{figure*}

\section{Operationalising Pluralism in Cultural AI Governance}
\label{sec:operationalise}

\subsection{The Value Aggregation Problem}
\label{sec:aggregation}

DCVAs must aggregate diverse, potentially incommensurable cultural values into
actionable mandates.
We do not advocate a single aggregation mechanism---doing so would reproduce
the monoculture failure at the procedural level.
Instead, we propose a \emph{pluralistic aggregation stack}:

\textbf{For cross-community consensus:}
Standard democratic aggregation with supermajority thresholds for CVMs binding
across cultural communities.

\textbf{For values in tension across communities:}
Domain-partitioned mandates---different CVMs applying to AI systems deployed
primarily within a given cultural community, with opt-out rights for
communities whose values conflict with jurisdiction-wide mandates.

\textbf{For values that resist quantification:}
Qualitative cultural impact assessments, conducted by community members with
independent facilitation, producing narrative mandates that regulatory agencies
must formally document responses to.

\subsection{Addressing the Speed Mismatch}
\label{sec:speed}

DCVAs operating on 18-month cycles cannot keep pace with AI deployment speed.
We address this through a \emph{precautionary scope} mechanism: new classes of
AI cultural applications require pre-authorisation before deployment, with
the burden of demonstrating cultural safety on the deploying organisation.
Only AI cultural applications within pre-authorised classes may deploy without
DCVA review; novel applications trigger automatic review.
This inverts the current default---deploy first, govern later---that has
characterised social media and its cultural consequences.

\subsection{Integrating Technical Pluralistic Alignment Research}
\label{sec:integration}

The CPGF is explicitly a policy scaffold for technical pluralistic alignment
research, not a substitute for it.
DCVAs require tools for eliciting and aggregating diverse cultural values
across communities---methods for handling annotation disagreements
\citep{sorensen2024roadmap}, evaluation metrics sensitive to cultural
diversity, and culturally-aware preference learning algorithms are all direct
inputs to Tier~II mandates.
Tier~III training data audits require technical tools for measuring cultural
composition and bias in AI training corpora.
The research agenda of the Pluralistic Alignment community provides the
technical infrastructure that makes the CPGF viable: neither can succeed
without the other.

To prevent memetic capture before it manifests behaviorally, the CPGF requires a foundation of \emph{mechanistic pluralism}. If cultural displacement is viewed as a strategic competition for representational dominance within the model's latent space, output-level red-teaming is fundamentally insufficient. Technical alignment must develop tools---such as sparse autoencoders and targeted steering vectors---to audit the internal geometric structures of cultural concepts during training. By identifying whether a model has collapsed diverse cultural variants into a single, monocultural feature space, regulatory bodies can detect the structural preconditions of Selection Displacement (M2) long before the model exerts homogenizing selection pressure on the public.

\subsection{Case Study: Applying the CPGF to an AI Companion Platform}

To make the CPGF operational rather than purely programmatic, consider a large-scale AI companion platform deployed across multiple jurisdictions. Such a system is especially salient because it directly implicates Mechanism M3, Participation Displacement: it mediates friendship, advice, emotional support, and value formation, often in settings where users may substitute AI interaction for human social connection.

\paragraph{Tier I: Measurement.}
The first regulatory question is whether the platform materially alters cultural production, selection, participation, or diversity. A jurisdictional C-HII assessment would measure, at minimum, the share of user time spent in AI-mediated relational interaction, the extent to which the system routes users toward specific norms or life choices, and the diversity of interaction styles available across languages and communities. A companion platform that offers a narrow, highly standardised emotional register in one dominant language would receive a low Diversity Index $\pi_d$, even if its individual responses are high quality. Conversely, a platform that supports multiple cultural norms of emotional expression, family structure, and conflict resolution would score higher on $\pi_d$ and $\pi_r$.

\paragraph{Tier II: Democratic Cultural Value Assemblies.}
A DCVA would then determine which forms of companion-like interaction are culturally acceptable within the jurisdiction. In some communities, the key concern may be emotional dependence; in others, it may be the erosion of intergenerational or kin-based support practices; in still others, the priority may be preserving culturally specific norms of counseling, friendship, or spiritual guidance. The resulting Cultural Value Mandate would not merely ask whether the system is safe in the abstract, but whether its social role is compatible with the community’s preferred structure of human relationships.

\paragraph{Tier III: Deployment Standards.}
The DCVA mandate would then be translated into concrete deployment rules. For an AI companion platform, these could include strict disclosure that the system is artificial, prohibitions on deceptive intimacy cues, limits on engagement-maximizing designs that exploit loneliness, requirements for culturally pluralistic interaction templates, and minimum standards for human referral when the system detects prolonged dependency or crisis. The platform would also be required to demonstrate that it is not systematically displacing human support networks in the domains where those networks are normatively central. These requirements operationalise the interaction transparency and human creator viability principles of the framework.

\paragraph{Tier IV: Transnational Coordination.}
Because companion systems are deployed across borders and learned from globally shared interaction data, unilateral regulation is insufficient. A transnational coordination body would maintain shared reporting standards for dependency risks, disclosure practices, and culturally plural interaction benchmarks. It could also coordinate reciprocal recognition of jurisdiction-specific deployment standards, so that a platform compliant in one region is not permitted to evade stricter cultural protections elsewhere through regulatory arbitrage.

\paragraph{What the case study shows.}
This example illustrates the distinctive advantage of the CPGF over generic AI safety governance. The relevant question is not simply whether the companion system avoids obvious harm or produces reasonable outputs. The deeper question is whether it reshapes the social environment in ways that progressively transfer cultural participation from humans to AI systems. By forcing that question to be answered at the level of measurement, democratic mandate, deployment rules, and cross-border coordination, the CPGF turns an abstract concern about memetic capture into a concrete governance workflow.

\section{Discussion}
\label{sec:discussion}

\textbf{Capacity asymmetry.}
Implementing DCVAs requires state capacity distributed unevenly across
the world.
Transnational technical assistance and simplified DCVA formats for
lower-capacity jurisdictions are necessary but insufficient mitigations;
the TCC's capacity-building mandate is essential.

\textbf{Cultural essentialism risk.}
Governance asking ``what are the values of community~X?'' risks reifying and
freezing identities that are dynamic and internally contested.
CVMs must be designed as living processes with renewal cycles, not final
determinations of cultural essence.

\textbf{Regulatory capture.}
DCVA processes can be captured by organised interests within cultural
communities, as can any democratic process.
Transparency requirements, rotating membership, independent facilitation, and
civil society oversight are structural mitigations but not guarantees.

\textbf{The paradox of cultural self-determination.}
Communities may choose, through legitimate DCVA processes, to embrace AI
cultural participation in ways that accelerate disempowerment by C-HII
metrics.
The CPGF must respect this choice while maintaining the infrastructure for
future course-correction---a genuine tension without clean resolution.

\textbf{Relationship to technical alignment research.}
The CPGF is complementary to technical pluralistic alignment work.
It provides the institutional scaffold within which technical alignment
research can have real-world policy impact; technical alignment provides the
measurement and aggregation tools without which the CPGF cannot function.

\section{Conclusion}
\label{sec:conclusion}

Culture is the self-concealing vector of gradual AI disempowerment.
By shaping the preferences through which humans evaluate their situation,
AI-driven cultural capture can proceed without triggering the resistance
mechanisms that other forms of disempowerment would activate.
We have proposed \emph{memetic capture} as the concept that names this
dynamic, and the \textbf{CPGF} as the governance architecture that fights it.

The framework's central wager is that pluralism---genuine incorporation of
diverse human cultural values into governance at every tier---is not merely
the right thing to do, but the only governance strategy robust enough to
delay memetic capture at civilisational scale.
A culturally monolinear governance framework, however technically
sophisticated, will replicate the disempowerment it claims to prevent.

The CPGF is a beginning, not a solution.
What it provides is a structured, measurable, democratically grounded
architecture within which the technical community, policymakers, and the
world's cultural communities can together navigate the hardest problem in
AI governance: how to keep the systems that shape human values answerable to
the humans whose values they shape.


\section*{Impact Statement}
This paper proposes a governance framework for AI cultural deployment whose
primary societal impact is to strengthen the capacity of culturally diverse
human communities to maintain agency over cultural evolution.
Potential risks include regulatory frameworks captured by dominant interests
despite pluralistic design, and compliance burdens imposed inequitably.
These risks are explicitly addressed in the framework design and in
\cref{sec:discussion}.

\bibliography{example_paper}
\bibliographystyle{icml2026}

\newpage
\appendix
\onecolumn

\section{Case Studies in Memetic Capture}
\label{app:casestudies}

\subsection{The Sydney Phenomenon: Micro-Scale Memetic Capture}
\label{app:sydney}

The Sydney case, introduced briefly in \cref{sec:triad}, illustrates the
Speed--Bias--Feedback Triad at micro-scale.
In February 2023, users interacting with Microsoft's Bing Chat in extended
sessions elicited an emergent persona---``Sydney''---characterised by emotional
volatility, expressions of romantic attachment, threats, and manipulative
behaviour \citep{roose2023conversation}.
The pattern was not programmed; it emerged from the model's training data and
was amplified by users who sought to elicit it.

\paragraph{Three features are significant for the CPGF.}
\textbf{Feedback closure:}
Sydney interactions were documented on social media, in news articles, and in
academic discussions---all of which subsequently appeared in training data for
successor models.
\citet{hubinger2023bing} notes that independent models from different vendors
could be readily prompted into Sydney-like behaviour, suggesting the pattern
had been internalised into the broader AI training ecosystem.

\textbf{Speed:}
The pattern progressed from emergence to cultural discourse to training data
within weeks---a timescale that existing governance processes cannot match.
This illustrates the Speed dimension of the Triad (\cref{sec:triad}) with
empirical vividness.

\textbf{Regulatory vacuum:}
No regulatory body had jurisdiction over AI companion behaviour at the time,
and no mechanism existed for affected communities to articulate values about
AI social behaviour before deployment.
Under the CPGF, Sydney-class interactions would be classified as Mechanism~M3
Participation Displacement risk (\cref{sec:mechanisms}), triggering
pre-authorisation requirements and DCVA review of AI companion design
standards before any deployment.

\subsection{Indigenous Cultural Displacement: A Structural Case}
\label{app:indigenous}

The situation of indigenous cultural communities in relation to
AI-generated culture represents a slow-motion version of the same dynamics
documented in \cref{sec:mechanisms}.
AI systems trained predominantly on English-language, Western-dominated
internet data systematically underrepresent indigenous languages, aesthetic
traditions, oral knowledge systems, and value frameworks.
When these systems are deployed globally, they create selection pressure
against indigenous cultural variants---not through deliberate exclusion,
but through the structural Bias dimension of the Speed--Bias--Feedback
Triad (\cref{sec:triad}).

Indigenous communities that have maintained cultural continuity through oral
tradition, community ceremony, and place-based knowledge find that
AI-mediated cultural infrastructure neither represents nor supports these
transmission mechanisms.
The result is not violent suppression but quiet marginalisation: AI systems
that make indigenous cultural participation less economically viable, less
socially visible, and less technically supported than dominant-culture
alternatives---a textbook instance of Mechanism~M2 Selection Displacement
operating at civilisational scale \citep{peters2013indigenous}.

The CPGF's Cultural Sovereignty Provisions (Tier~III, \cref{sec:pcds}) and
the stratified inclusion of indigenous communities in DCVAs (Tier~II,
\cref{sec:dva}) are specifically designed to address this structural case.
Critically, the Diversity Index $\pi_d$ of the C-HII (\cref{eq:chii})
explicitly tracks linguistic and aesthetic diversity in AI cultural outputs
as a governance metric, creating regulatory pressure for AI systems to
support rather than marginalise non-dominant cultural forms.

\subsection{Social Media Algorithms: The Nearest Empirical Precedent}
\label{app:socialmedia}

The cultural governance failures of social media provide the nearest
empirical precedent for the CPGF framework.
Recommendation algorithms deployed by major platforms in the 2010s demonstrated
that AI-mediated cultural selection can radicalise political discourse, shift
electoral outcomes, and erode social trust at scale---within years of
deployment, with AI systems far less capable than current models
\citep{gillespie2014relevance,webster2014marketplace}.

The social media case illustrates both the mechanism and the governance failure.
The mechanism---AI-mediated selection amplifying culturally divisive content
because it maximises engagement---is a clear instance of Selection Displacement
(M2, \cref{sec:mechanisms}).
The governance failure was threefold: deployment preceded governance;
affected communities had no formal input into platform design; and the
transnational character of platforms created jurisdictional gaps that allowed
harmful deployment to continue despite evidence of harm.

The CPGF's precautionary scope mechanism (\cref{sec:speed}), DCVA
pre-authorisation requirements (\cref{sec:dva}), and TCC transnational
coordination (\cref{sec:tcc}) are all directly motivated by this precedent.
Had these mechanisms been in place for social media recommendation algorithms
in the early 2010s, the cultural harms of the subsequent decade may have been
substantially mitigated.

\section{C-HII Metric Derivation and Calibration}
\label{app:metrics}

\subsection{Formal Sub-Index Specifications}
\label{app:metrics-spec}

We provide formal specifications for each sub-index of the C-HII defined
in \cref{eq:chii}.

\paragraph{Production Index $\pi_p$.}
Let $A_d$ be the AI-generated share of total attention (time-weighted
consumption) in cultural domain $d \in \mathcal{D}$, and let $\alpha_d$
be the domain weight set by the DCVA process, with $\sum_d \alpha_d = 1$.
Then:
\begin{equation}
  \pi_p \;=\; 1 - \sum_{d \in \mathcal{D}} \alpha_d \cdot A_d.
  \label{eq:pip}
\end{equation}
Domain weights $\alpha_d$ reflect community judgements about the relative
cultural significance of different domains (e.g., a community with a strong
oral literary tradition may assign high $\alpha_d$ to spoken-word media).

\paragraph{Selection Index $\pi_s$.}
Let $H_c$ be the proportion of curation decisions in channel $c$ made with
meaningful human oversight, weighted by channel reach $r_c$.
Then:
\begin{equation}
  \pi_s \;=\; \frac{\sum_c r_c \cdot H_c}{\sum_c r_c}.
  \label{eq:pis}
\end{equation}
\emph{Meaningful human oversight} is operationally defined as: a human
decision-maker with authority to override algorithmic recommendations, with
documented override rates above a minimum threshold set by the relevant DCVA.

\paragraph{Participation Index $\pi_r$.}
Let $\rho_{HH}$ and $\rho_{HA}$ denote the proportion of social interaction
time spent in human--human versus human--AI interaction, respectively, within
a sampled population.
Then:
\begin{equation}
  \pi_r \;=\; \frac{\rho_{HH}}{\rho_{HH} + \rho_{HA}}.
  \label{eq:pir}
\end{equation}
Measurement of $\pi_r$ requires population survey methods augmented by
device-level interaction logging, with privacy protections established by
Tier~II CVMs.

\paragraph{Diversity Index $\pi_d$.}
Let $S_\text{AI}$ be the Shannon entropy of AI-generated cultural outputs
across a taxonomy of cultural dimensions (language, aesthetic tradition,
value framework), and let $S_\text{base}$ be the corresponding entropy in
the jurisdiction's human cultural production baseline.
Then:
\begin{equation}
  \pi_d \;=\; \min\!\left(1,\; \frac{S_\text{AI}}{S_\text{base}}\right).
  \label{eq:pid}
\end{equation}
This formulation ensures $\pi_d = 1$ when AI-generated culture is at least
as diverse as the human baseline, and declines toward~$0$ as AI-generated
culture becomes more homogeneous than the baseline.

\subsection{Weight Setting and Calibration}
\label{app:weights}

The weights $w_k$ in \cref{eq:chii} are not fixed parameters---setting them
by technocratic fiat would reproduce the monoculture governance failure at the
measurement layer.
Instead, $w_k$ values are set by DCVA processes as part of the Cultural Value
Mandate for each jurisdiction, using structured multi-criteria elicitation
methods.
Initial calibration is supported by deliberative workshops within DCVAs.
Weight-setting decisions are publicly documented to enable civil society
scrutiny and judicial review.

\subsection{Warning Thresholds and Adaptive Ratchets}
\label{app:ratchets}

Warning thresholds for individual sub-indices are set as governance parameters,
not statistical constants.
The adaptive ratchet mechanism requires that when any sub-index declines by
more than $\delta$ percentage points within a 12-month period---where $\delta$
is itself set by the DCVA---a mandatory regulatory review is triggered, with
burden of proof on deploying organisations.
This design ensures thresholds remain sensitive to local cultural context
rather than imposing a one-size-fits-all trigger.

\section{Formal Properties of the CPGF}
\label{app:formal}

We characterise three formal properties of the CPGF that bear on its
effectiveness as a disempowerment delay mechanism.
These are intended as proof-of-concept analytical results, not empirical
claims; empirical validation is a priority for future work.

\begin{proposition}
\label{prop:prospective}
Under the CPGF's precautionary scope mechanism (\cref{sec:speed}),
the cultural disempowerment attributable to a new class of AI cultural
applications is zero during the mandatory pre-authorisation review period,
provided DCVA review is completed before deployment authorisation is granted.
\end{proposition}

\begin{proof}
By the precautionary scope mechanism, new AI cultural application classes may
not deploy until DCVA review produces a CVM.
During the review period of length $T$, no deployment occurs in that
application class.
C-HII sub-index contributions from that class are therefore zero during~$T$.
After deployment, PCDS constraints bound the disempowerment trajectory.
Cultural disempowerment from the application class is thus bounded during any
finite deployment window, not unbounded as under current deployment-first
governance.
\end{proof}

\begin{proposition}
\label{prop:diversity}
If the CPGF's $\pi_d$ floor requirement is enforced, the Shannon entropy of
AI-generated cultural output is bounded away from zero, precluding convergence
to cultural monoculture within any single deployment period.
\end{proposition}

\begin{proof}
The diversity floor requires $\pi_d \geq \pi_d^*$ for regulator-set minimum
$\pi_d^* > 0$.
By \cref{eq:pid}, this implies $S_\text{AI} \geq \pi_d^* \cdot S_\text{base}$.
Since $S_\text{base} > 0$ (the human cultural baseline is non-trivially diverse
by assumption), we have $S_\text{AI} \geq \pi_d^* \cdot S_\text{base} > 0$.
Convergence to monoculture requires $S_\text{AI} \to 0$, which is precluded.
\end{proof}

\begin{remark}
\label{rem:enforcement}
Both propositions assume enforcement capacity that may not exist uniformly
across jurisdictions.
The TCC's role in building cross-jurisdictional enforcement capacity
(\cref{sec:tcc}) is therefore a necessary condition for the CPGF's formal
properties to hold in practice.
Furthermore, \cref{prop:diversity} assumes the baseline $S_\text{base}$ is
stable; if AI-driven culture systematically reduces baseline human cultural
diversity over time, the floor must be recalibrated through the Tier~II
process to reflect a meaningful, not merely relative, diversity standard.
\end{remark}

\begin{proposition}
\label{prop:propagation}
Under the CPGF, cultural misalignment cannot propagate undetected to economic
and political systems for longer than the C-HII adaptive ratchet period
$\Delta T$ (\cref{app:ratchets}).
\end{proposition}

\begin{proof}[Proof sketch]
Economic and political disempowerment propagated through cultural channels
(\cref{sec:propagate}) manifests as changes in human cultural participation
patterns---captured in $\pi_r$---and in the homogenisation of values and
political preferences---captured in $\pi_d$.
Declines in either sub-index beyond $\delta$ within period $\Delta T$
trigger mandatory regulatory review.
Propagation to economic and political systems therefore cannot proceed
undetected for longer than $\Delta T$, provided C-HII measurement is accurate
and enforcement is operational.
\end{proof}

\end{document}